\documentclass[10pt]{article}
\usepackage{color,graphicx}               
\usepackage{amsmath,amssymb,amsthm,url}
\usepackage[numbers,sort&compress]{natbib}
\usepackage{newcent}
\def\d{{\rm d}}

\def\i{\mbox{\rm i}}
\newcommand{\deriv}[3][]{\frac{\d^{#1}{#2}}{{\d{#3}}^{#1}}}

\def\a{\alpha}
\def\b{\beta}

\def\ph{\varphi}
\def\th{\vartheta}

\def\p{Painlev\'{e}}
\def\peq{\p\ equation}
\def\peqs{\p\ equations}
\def\sch{Schr\"{o}dinger}
\def\bt{B\"{a}cklund transformation}
\def\bts{B\"{a}cklund transformations}

\def\Ep{Ermakov-\p}

\def\PI{${\rm P}_{\rm{I}}$}
\def\PII{${\rm P}_{\rm{II}}$}
\def\PIV{${\rm P}_{\rm{IV}}$}
\def\PVI{${\rm P}_{\rm{VI}}$}

\def\SIV{${\rm S}_{\rm{IV}}$}
\def\erfc{\mathop{\rm erfc}\nolimits}

\def\Cases#1{\begin{cases}#1\end{cases}}
\newtheorem{theorem}{Theorem}[section]

\newtheorem{lemma}[theorem]{Lemma}
\newtheorem{corollary}[theorem]{Corollary}
\theoremstyle{definition}

\newtheorem{remark}[theorem]{Remark}

\def\beq{\begin{equation}}
\def\eeq{\end{equation}}
\newcommand{\hide}[1]{}
\def\etal{\textit{et al.}}
\numberwithin{equation}{section}
\numberwithin{theorem}{section}
\numberwithin{figure}{section}

\begin{document}
	
	\title{On Integrable Ermakov-Painlev\'{e} IV Systems}
	
	\author{Colin Rogers~$^1$  Andrew P. Bassom${}^2$ and Peter A.\ Clarkson~$^3$\\[5pt]
$^1$~School of Mathematics \& Statistics, The University of New South Wales, \\Sydney, NSW2052, Australia\\
\texttt{c.rogers@unsw.edu.au} \\[5pt]
$^2$~School of Mathematics and Physics, University of Tasmania,\\ Private Bag 37, Hobart, Tasmania 7001, Australia\\
\texttt{andrew.bassom@utas.edu.au}\\[5pt]
$^3$~School of Mathematics, Statistics \& Actuarial Science,\\ University of Kent, Canterbury, CT2 7FS, UK\\
\texttt{P.A.Clarkson@kent.ac.uk}} 
	
	\maketitle
	
\begin{abstract}
Novel hybrid Ermakov-Painlev\'{e} IV systems are introduced and an associated Ermakov invariant is used in establishing their integrability. B\"{a}cklund transformations are then employed to generate classes of exact solutions via the linked canonical Painlev\'{e} IV equation.
\end{abstract}

\section{Introduction}

The six classical \p\ equations, commonly denoted by \PI--\PVI, arise in a wide range of physical applications and have played a fundamental role in modern soliton theory. Detailed accounts of their properties may be found in \cite{pac05,pac06,rc99,refGLS02} 
together with the literature cited therein. In particular, the fourth \p\ equation (\PIV) which is connected to the new integrable systems to be introduced here adopts the canonical form
\begin{equation} \label{eq:PIV}
 \deriv[2]{\omega}{z} = \frac1{2\omega}\left(\deriv{\omega}{z}\right)^2 + \frac{3}{2}\omega^3 + 4z \omega^2 + 2 (z^2 - \a) \omega + \frac{\b}{\omega},
\end{equation}
where $\a$ and $\b$ are {in general arbitrary parameters, possibly complex. In the present context they are taken as real}. This nonlinear integrable equation has been shown to have important applications such as hydrodynamics \cite{bch96}, nonlinear optics \cite{refFG}, quantum gravity theory \cite{refFIKa,refFIKb}, 
supersymmetric quantum mechanics \cite{refBerm12,refBCF14,refBF11,refBF12,refFG15,refMar09,refMateo}
random matrices \cite{refFW01,refFW03,refKanz,refOK08,refTW94}, 
orthogonal polynomials \cite{refCF06,refCP05,refCJ14,refCJK,refDK09},
exceptional orthogonal polynomials \cite{refGG,refGGM,refGKKM,refMarQ13,refMarQ16} 
and vortex dynamics \cite{pac09}.
Associated with \PIV\ \eqref{eq:PIV} is the second-degree, second-order $\sigma$-equation, usually known as \SIV\ (or the \PIV\ $\sigma$-equation)
\beq\label{eq:Siv}
\left(\deriv[2]{\sigma}z\right)^{2} - 4\left(z\deriv{\sigma}z-\sigma\right)^{2} +4\deriv{\sigma}z\left(\deriv{\sigma}z+2\th_0\right)\left(\deriv{\sigma}z+2\th_{\infty}\right)=0,
\eeq
where $\th_0$ and $\th_{\infty}$ are arbitrary parameters, due to Jimbo and Miwa \cite{refJMi} and Okamoto \cite{refOkamoto80a,refOkamoto80b,refOkamotoPIIPIV}, whose solutions are equivalent to those of \PIV\ \eqref{b8}.

Nonlinear coupled systems of Ermakov-Ray-Reid type, on the other hand, namely (see e.g.\ \cite{jr80,jrjr80,crchjr93,crws96,wscrab96})
\begin{subequations}\label{a2}\begin{align}
 \deriv[2]{\phi}{z} + \Theta(z) \phi &= \frac{1}{\phi^2 \psi}\ \Phi (\psi/\phi), \\
 \deriv[2]{\psi}{z}+ \Theta(z) \psi &= \frac{1}{\psi^2 \phi}\ \Psi (\phi/\psi),
 \end{align}
\end{subequations}
with their distinctive integral of motion
\begin{equation*} 
\mathcal{I} = \tfrac12\left( \phi \deriv{\psi}{z} - \psi \deriv{\phi}{z}\right)^2 + \int^{u=\psi/\phi} \Phi(u) \,\d u + \int^{v=\phi/\psi} \Psi(v) \,\d v,
 \end{equation*}
likewise display a diverse range of physical applications as described in \cite{agylaspsst91,crha10,crbmha12,crbmkcha10,crws11,crws2011,wshacr13}. 

It is well known that solitonic systems generically admit invariance under \bts\ which have associated nonlinear superposition principles{, namely,} permutability theorems {(see e.g. \cite{crws82,crws02})}. \bts\ generically admitted by the \p\ equations such as set down in \cite{vg99} enable the iterative generation of sequences of exact solutions of associated overarching solitonic equations. Ermakov-Ray-Reid systems likewise possess nonlinear superposition laws, albeit of another kind \cite{jr80,jrjr80}. Moreover, just as solitonic systems and their associated \p\ symmetry reductions admit linear representations, so Ermakov-Ray-Reid systems have also been shown to possess underlying linear structure in \cite{cacrurao90}. Despite these remarkable commonalities, the studies of \p\ and Ermakov-type systems have preceded independently until recently.

{The notions of S-integrable and C-integrable nonlinear equations with their important associated universal aspects is due to Calogero (see e.g. \cite{fc93} and literature cited therein). The conjugation of such S-integrable equations which are amenable to the inverse scattering transform and nonlinear C-integrable equations which may be reduced to tractable linear canonical forms, typically via B\"acklund, gauge or reciprocal transformations, constitutes a novel topic with origin in work on generalised Ermakov systems as introduced in \cite{wscrab96}. Therein, it was established that an S-integrable 2+1-dimensional extension due to Schief \cite{ws94} of the Ernst equation of general relativity, incorporates a C-integrable generalised Ermakov system. The present work has its genesis in recent developments in \cite{cr14,cr2014}, where prototype integrable hybrid Ermakov-Painlev\'e II and Ermakov-Painlev\'e IV-type systems were derived via wave packet symmetry reductions of two classes of physically important resonant nonlinear Schr\"odinger systems.} The latter involve de Broglie-Bohm quantum potential terms and under certain circumstances they admit novel fission or fusion phenomena \cite{jlopcrws07,opjl02,opjlcr08}.

Dirichlet two-point boundary value problems for both the \Ep\ II and \Ep\ IV equations have been recently investigated in \cite{pacr15,pacr2015}. In \cite{pacr15}, the hybrid \Ep\ II equation was linked to its integrable \p\ XXXIV avatar 
in the context of a three-ion reduction of the classical Nernst-Planck system as derived in \cite{rccrws07}. On setting $\Omega=\omega^{1/2}$, with $\omega>0$,
in \PIV\ \eqref{eq:PIV}, one obtains the hybrid \Ep\ IV type equation
\begin{equation} \label{a4}
\deriv[2]{\Omega}{z}= \left[ \tfrac34\Omega^4 + 2z \Omega^2 + z^2 - \a \right] \Omega+ \frac{\b}{2\Omega^3}, \end{equation}
which encapsulates positive solutions of \PIV\ \eqref{eq:PIV}. It may be derived, in particular, via symmetry reduction of derivative nonlinear \sch\ equations of a type that arise in plasma physics, cf.~\cite{refCC87}. 
The Ermakov invariant admitted by coupled pairs of nonlinear equations of the type (\ref{a4}) has been shown in \cite{pacr15} to allow their integration in terms of \PIV\ \eqref{eq:PIV}. 

In the present paper a novel Ermakov-Ray-Reid type system is introduced and its integrable structure established via a Ermakov invariant relation together with a linked canonical \PIV\ equation.

\section{Integrability of the \Ep\ IV System}

The preceding motivates the introduction here of the hybrid \Ep\ IV system 
\begin{align*}
\deriv[2]{\Phi_1}{z} - \left[ \tfrac34(\Phi^2_1 + \Phi^2_2)^2 + 2z ( \Phi^2_1 + \Phi^2_2 ) + z^2 - \a \right] \Phi_1 &= \frac{1}{\Phi^2_1 \Phi_2} S \left(\frac{\Phi_2}{\Phi_1}\right), \\
\deriv[2]{\Phi_2}{z} - \left[ \tfrac34(\Phi^2_1 + \Phi^2_2)^2 + 2z (\Phi^2_1 + \Phi^2_2) + z^2 - \a \right] \Phi_2 &= \frac{1}{\Phi^2_2 \Phi_1} T \left(\frac{\Phi_1}{\Phi_2}\right). 
 \end{align*}
The integrability of this coupled nonlinear system is investigated when there exists $V(\Phi_1, \Phi_2)$ such that
\begin{equation*} 
\frac{1}{\Phi^2_1 \Phi_1} S \left(\frac{\Phi_2}{\Phi_1}\right) = \frac{\partial V}{\partial \Phi_1} ,\qquad \frac{1}{\Phi^2_2 \Phi_1} T \left(\frac{\Phi_1}{\Phi_2}\right) = \frac{\partial V}{\partial \Phi_2}, \end{equation*}
in which case, it may be shown, as in the manner described in \cite{crha10} for the standard Ermakov-Ray-Reid system, that it adopts the form
\begin{subequations} \label{b3}
\begin{align} \label{b3a}
\deriv[2]{\Phi_1}{z} &- \left[ \tfrac34(\Phi^2_1 + \Phi^2_2)^2 + 2z ( \Phi^2_1 + \Phi^2_2 ) + z^2 - \a \right] \Phi_1
= \frac{2}{\Phi^3_1} J \left(\frac{\Phi_2}{\Phi_1}\right) + \frac{\Phi_2}{\Phi^4_1} J' \left(\frac{\Phi_2}{\Phi_1}\right), \\
\deriv[2]{\Phi_2}{z} &- \left[ \tfrac34(\Phi^2_1 + \Phi^2_2)^2 + 2z (\Phi^2_1 + \Phi^2_2) + z^2 - \a \right] \Phi_2 = - \frac{1}{\Phi^3_1} J' \left(\frac{\Phi_2}{\Phi_1}\right), \label{b3b}
 \end{align}
\end{subequations}
%
%
where the prime denotes a derivative with respect to the argument $\Phi_2/\Phi_1$ in $J$. This system admits the Ermakov invariant
\begin{equation} \label{b4}
\mathcal{I} = \tfrac12\left(\Phi_1 \deriv{\Phi_{2}}{z} - \deriv{\Phi_{1}}{z} \Phi_2\right)^2 + \left( \frac{\Phi^2_1 + \Phi^2_2}{\Phi^2_1} \right) J \left(\frac{\Phi_2}{\Phi_1}\right), \end{equation}
and, on use of the identity
\begin{align*} \label{b5}
(\Phi^2_1 + \Phi^2_2) \left[\left(\deriv{\Phi_{1}}{z}\right)^2 + \left(\deriv{\Phi_{2}}{z}\right)^2\right] &- \left(\Phi_1 \deriv{\Phi_{2}}{z} - \deriv{\Phi_{1}}{z} \Phi_2\right)^2 
\equiv \left(\Phi_1 \deriv{\Phi_{1}}{z} + \Phi_2 \deriv{\Phi_{2}}{z}\right)^2, \end{align*}
it follows that
\begin{equation*}
\left(\deriv{\Phi_{1}}{z}\right)^2 + \left(\deriv{\Phi_{2}}{z}\right)^2 - \frac{2 \mathcal{I}}{\Sigma} + \frac{2}{\Phi^2_1} J \left(\frac{\Phi_2}{\Phi_1}\right) = \frac{1}{4\Sigma}\left(\deriv{\Sigma}{z}\right)^2 ,\end{equation*}
where $\Sigma\equiv\Phi^2_1+\Phi^2_2$. Hence
\begin{align} 
2\left( \deriv{\Phi_{1}}{z} \deriv[2]{\Phi_1}{z} + \deriv{\Phi_{2}}{z} \deriv[2]{\Phi_2}{z}\right) &+ 
\frac{2 \mathcal{I}}{\Sigma^2}\deriv{\Sigma}{z} + \deriv{}{z}\left[ \frac{2}{\Phi^2_1} J \left(\frac{\Phi_2}{\Phi_1}\right) \right]
= \frac{1}{2\Sigma}\deriv{\Sigma}{z}\deriv[2]{\Sigma}{z} - \frac{1}{4\Sigma^2}\left(\deriv{\Sigma}{z}\right)^3, \label{b6}\end{align}
where the system \eqref{b3} shows that
\begin{align*}
\deriv{\Phi_{1}}{z} \deriv[2]{\Phi_1}{z} + \deriv{\Phi_{2}}{z} \deriv[2]{\Phi_2}{z} - \tfrac12 \Delta \deriv{\Sigma}{z} &= \frac{2}{\Phi^3_1} \deriv{\Phi_{1}}{z}J + \left(\frac{\Phi_2}{\Phi^4_1} \deriv{\Phi_{1}}{z} - \frac{1}{\Phi^3_1} \deriv{\Phi_{2}}{z}\right)J 
= - \deriv{}{z}\left[ \frac{1}{\Phi^2_1} J \left(\frac{\Phi_2}{\Phi_1}\right) \right],
\end{align*}
with $\Delta: = \tfrac34\Sigma^2 + 2 \Sigma + z^2 - \a$. Accordingly, (\ref{b6}) shows that reduction is obtained to the canonical \PIV\ equation in $\Sigma$, namely
\begin{equation} \label{b8}
\deriv[2]{\Sigma}{z} = \frac{1}{2\Sigma}\left(\deriv{\Sigma}{z}\right)^2 + \frac{3}{2} \Sigma^3 + 4z \Sigma^2 + 2 (z^2 - \a) \Sigma + \frac{\b}{\Sigma}, \end{equation}
with $\b=4\mathcal{I}$.

To determine the original Ermakov variables $\Phi_1$ and $\Phi_2$ corresponding to a known positive solution $\Sigma$ of \PIV\ \eqref{b8}, one returns to the Ermakov invariant relation (\ref{b4}). Thus, on introduction of $\Lambda$ according to
\begin{equation} \label{b9}
\Lambda = \frac{2 \Phi_1 \Phi_2}{\Phi^2_1 + \Phi^2_2}, \end{equation}
it is observed that
\begin{equation*}
\deriv{\Lambda}{z} = \frac{2 (\Phi^2_1 - \Phi^2_2)}{(\Phi^2_1 + \Phi^2_2)^2}\left(\Phi_1 \deriv{\Phi_{2}}{z} - \deriv{\Phi_{1}}{z} \Phi_2\right), \end{equation*}
whence (\ref{b4}) yields
\begin{equation} \label{b11}
\mathcal{I} = \tfrac{1}{8}\ \frac{(\Phi^2_1 + \Phi^2_2)^4}{(\Phi^2_1 - \Phi^2_2)^2}\left(\deriv{\Lambda}{z}\right)^2 + \left( \frac{\Phi^2_1 + \Phi^2_2}{\Phi^2_1} \right) J (\Phi_2 / \Phi_1).\end{equation}
Hence, in terms of the new independent variable $z^*$ where
\begin{equation} \label{b12}
\d z^* = \Sigma^{-1}\,\d z, \end{equation}
invariant relation (\ref{b11}) shows that
\begin{equation*}
\Lambda^2_{z^*} = 8 \left( \frac{\Phi_1^2- \Phi_2^2}{\Phi_1^2 + \Phi_2^2} \right)^2 \left[\mathcal{I} - \frac{\Phi_1^2 + \Phi_2^2}{\Phi_1^2 }\,J \left(\frac{\Phi_2}{\Phi_1}\right)\right], \end{equation*}
where (\ref{b9}) gives $\Phi_2/\Phi_1$ in terms of $\Lambda$ according to
\begin{equation*}
\Phi_2/\Phi_1 = ( 1 \pm \sqrt{1 - \Lambda^2} )/ \Lambda.\end{equation*}
Thus,
\begin{equation} \label{b15}
\pm \frac{1}{2\sqrt{2}} \int \sqrt{\frac{\Lambda}{(1 - \Lambda^2)(\Lambda \mathcal{I} - 2\mathcal{L} (\Lambda))}}\ d \Lambda = z^* + \mathcal{C}, \end{equation}
where $\mathcal{L}(\Lambda):=\left(\Phi_2/\Phi_1\right)J\left(\Phi_2/\Phi_1\right)$ and $\mathcal{C}$ is an arbitrary constant of integration. Corresponding to positive solutions $\Sigma$ of \PIV\ \eqref{b8} and $\Lambda$ as determined by (\ref{b15}) for specified $J\left({\Phi_2}/{\Phi_1}\right)$, the original variables $\Phi_1,\ \Phi_2$ in the \Ep\ IV system are given via the relations
\begin{equation} \label{b16}
\Phi^2_1 = \tfrac12\Sigma \left( 1 \mp \sqrt{1 - \Lambda^2} \right),\qquad \Phi^2_2 = \tfrac12\Sigma \left( 1 \pm \sqrt{1 - \Lambda^2} \right).\end{equation}

The class of \Ep\ IV systems will be considered here which corresponds to $\Lambda$ as determined by the quadrature in (\ref{b15}) with
\begin{equation} \label{b17}
\Lambda \mathcal{I} - 2 \mathcal{L} (\Lambda) = \frac{a^2 (1 - \Lambda^2)^m}{\Lambda}.\end{equation}
In this case, (\ref{b15}) yields
\begin{equation} \label{b18}
1 - \Lambda^2 = \Cases{ \left[\pm 2 \sqrt{2}\,a (m-1) (z^*+\mathcal{C})\right]^{2/(1-m)} ,&\mbox{\rm if}\quad m \neq 1, \\
\exp \left\{ \pm 4\sqrt{2}\,a (z^*+\mathcal{C}) \right\}, & \mbox{\rm if}\quad m = 1, } \end{equation}
where $z^*$ is obtained by integration of the relation (\ref{b12}) for a positive solution $\Sigma$ of \PIV. Here, (\ref{b17}) shows that
\begin{equation} \label{b19}
J \left(\frac{\Phi_2}{\Phi_1}\right) = \frac{\mathcal{I}}{1+\left(\Phi_2/\Phi_1\right)^2} - \frac{a^2}{4} \left( \frac{\Phi^2_1 - \Phi^2_2}{\Phi^2_1 + \Phi^2_2} \right)^{2m} \,\frac{\Phi_1^2 + \Phi_2^2}{\Phi_1^2 }.\end{equation}

\section{Action of B\"{a}ck\-lund-Type transformations}

The diversity of B\"{a}ck\-lund and Schlesinger-type transformations admitted by \PIV\ \eqref{b8} has been described in the comprehensive work of \cite{bch95}. Thus, transformations due to Gromak \cite{vg78,vg87} 
Fokas, Mugan and Ablowitz \cite{afumma88}, Lukashevich \cite{nl65,nl67} and Murata \cite{ym85} are all detailed therein along with their conjugation and interconnections. In the present context, the concern is to exploit certain such transformations to generate sequences of exact solutions to the \Ep\ IV system via the action on appropriate seed solutions. In \cite{bch95}, a multiplicity of combinations of the parameters $\a$ and $\b$ has been catalogued which correspond to exact solutions of \PIV\ \eqref{b8} . These may involve error functions, parabolic cylinder functions or rational functions. The generation of solutions of the \Ep\ systems corresponding to the three parameter class of $J(\Phi_1, \Phi_2)$ as given by (\ref{b19}) is illustrated below for elementary rational seeds.

In the above connection, it is seen that \PIV\ \eqref{b8} admits a seed solution
\begin{equation} \label{c1}
\Sigma_0(z) = -2z,\qquad \a = 0, \qquad \b = -2, \qquad \mathcal{I} = -\tfrac12, \end{equation}
where, in the present context the requirement that $\Sigma>0$ restricts attention to regions with $z<0$. With this seed solution, (\ref{b12}) shows that, up to an additive constant,
\begin{equation*}
z^* = - \tfrac12 \ln (-z), \end{equation*}
so that the \Ep\ IV system with $J\left(\Phi_2/\Phi_1\right)$ given by (\ref{b19}) admits the class of solutions with
\begin{align*}
\Phi^2_1 &= \Cases{- \tfrac12z\left( 1 \mp \left[ \pm 2 \sqrt{2}\,a (m-1) \left( - \tfrac12\ln (-z) + \mathcal{C} \right)
 \right]^{1/(1-m)} \right) ,&\mbox{\rm if}\quad m \neq 1, \\
- \tfrac12z\left( 1 \mp \exp \left[ \mp 4 \sqrt{2}\,a \left( - \tfrac12\ln (-z) + \mathcal{C} \right) \right] \right),& \mbox{\rm if}\quad m = 1, }\\
\Phi^2_2 &= \Cases{- \tfrac12z\left( 1 \pm \left[ \pm 2 \sqrt{2}\,a (m-1) \left( - \tfrac12\ln (-z) + \mathcal{C} \right) \right]^{1/(1-m)} \right)
 ,&\mbox{\rm if}\quad m \neq 1, \\
- \tfrac12z\left( 1 \pm \exp \left[ \mp 4\sqrt{2}\,a \left( - \tfrac12\ln(-z)+\mathcal{C} \right) \right] \right) ,&\mbox{\rm if}\quad m = 1, }
\end{align*}
and appropriate choices of sign made to ensure positivity requirements.

It is interesting to observe that for both the seed solutions of \PIV\ \eqref{b8} 
\[\begin{array}{l@{\qquad}l@{\qquad}l}
\Sigma_0(z)=-2z, & \a=0, & \b=-2,\\
\widetilde{\Sigma}_0(z)=-\tfrac23z, & \a=0, & \b=-\tfrac29,
\end{array}\]
the \Ep\ IV system \eqref{b3} reduces to a conventional Ermakov-Ray-Reid system
\begin{equation*}
\deriv[2]{\Phi_{1}}{z} = \frac{1}{\Phi^2_1 \Phi_2} S \left(\frac{\Phi_2}{\Phi_1}\right) ,\qquad \deriv[2]{\Phi_{2}}{z} = \frac{1}{\Phi^2_2 \Phi_1} T \left(\frac{\Phi_1}{\Phi_2}\right).\end{equation*}

\def\q{\omega}
{\begin{theorem}\label{thm:bts}Let $\q_0=\q(z;\a_0,\b_0)$ and $\q_j^{\pm}= \q(z;\a_j^{\pm},\b_j^{\pm})$,
$j=1,2,3,4$ be solutions of \PIV\ \eqref{eq:PIV} with
\begin{align*}
(\a_1^{\pm},\b_1^{\pm}) &= \left(\tfrac14(2-2\a_0 \pm 3\sqrt{-2\b_0}), -\tfrac12(1+\a_0 \pm \tfrac12\sqrt{-2\b_0})^{2}\right),\\[2.5pt]
(\a_2^{\pm},\b_2^{\pm}) &= \left(-\tfrac14(2+2\a_0 \pm 3\sqrt{-2\b_0}),-\tfrac12(1-\a_0 \pm \tfrac12\sqrt{-2\b_0})^{2}\right),\\[2.5pt]
(\a_3^{\pm},\b_3^{\pm}) &= \left(\tfrac32-\tfrac12\a_0\mp\tfrac34\sqrt{-2\b_0},-\tfrac12(1- \a_0\pm\tfrac12\sqrt{-2\b_0})^{2}\right),\\[2.5pt]
(\a_4^{\pm},\b_4^{\pm}) &= \left(\tfrac32-\tfrac12\a_0\mp\tfrac34\sqrt{-2\b_0},-\tfrac12(-1-\a_0 \pm\tfrac12\sqrt{-2\b_0})^{2}\right).\end{align*}
Then
\begin{subequations}
\begin{align}
\mathcal{T}_1^{\pm}:\qquad \q_1^{\pm} &=\frac{\q_0' -\q_0^2 -2z\q_0\mp\sqrt{-2\b_0}}{2\q_0},\\
\mathcal{T}_2^{\pm}:\qquad \q_2^{\pm} &=-\,\frac{\q_0' +\q_0^2+2z\q_0\mp\sqrt{-2\b_0}}{2\q_0},\\
\mathcal{T}_3^{\pm}:\qquad \q_3^{\pm} &=\q_0+\frac{2\left(1-\a_0\mp\tfrac12\sqrt{-2\b_0}\,\right)\q_0}{\q_0' \pm\sqrt{-2\b_0}+2z\q_0 +\q_0^2},\\
\mathcal{T}_4^{\pm}:\qquad \q_4^{\pm} &= \q_0+\frac{2\left(1+\a_0\pm\tfrac12\sqrt{-2\b_0}\,\right)\q_0} {\q_0'\mp\sqrt{-2\b_0}-2z\q_0 -\q_0^2},
\end{align}\end{subequations} valid when the denominators are non-zero,
and where the upper signs or the lower signs are taken throughout each
transformation.\end{theorem}

\begin{proof}See Gromak \cite{vg78,vg87} and Lukashevich \cite{nl65,nl67}; 
 also \cite{bch95,afma82,refGLS02,vgnl82,ym85}.
\end{proof}}

A class of \bts\ for the \peqs\ is generated by so-called \textit{Schlesinger transformations} of the associated isomonodromy problems.
Fokas, Mugan and Ablowitz \cite{afumma88} (see also \cite{mu92}), deduced the following Schlesinger transformations
$\mathcal{R}_1$--$\mathcal{R}_4$ for \PIV
\begin{subequations}
\begin{align}  \label{STR1}
\mathcal{R}_1:\quad \q_1(z;\a_1,\b_1)
&= \frac{\left(\q_0'+\sqrt{-2\b_0}\right)^2+\left(4\a+4-2\sqrt{-2\b_0}\right)\q_0^2-\q_0^2(\q_0+2z)^2}{2\q_0\left(\q_0^2+2z\q_0-\q_0'-\sqrt{-2\b_0}\right)},\\ \label{STR2}
\mathcal{R}_2:\quad \q_2(z;\a_2,\b_2)
&= \frac{\left(\q_0'-\sqrt{-2\b_0}\right)^2+\left(4\a-4-2\sqrt{-2\b_0}\right)\q_0^2-\q_0^2(\q_0+2z)^2}{2\q_0\left(\q_0^2+2z\q_0+\q_0'-\sqrt{-2\b_0}\right)},\\ \label{STR3}
\mathcal{R}_3:\quad \q_3(z;\a_3,\b_3) 
&= \frac{\left(\q_0'-\sqrt{-2\b_0}\right)^2+\left(4\a+4+2\sqrt{-2\b_0}\right)\q_0^2-\q_0^2(\q_0+2z)^2}{2\q_0\left(\q_0^2+2z\q_0-\q_0'+\sqrt{-2\b_0}\right)},\\ \label{STR4}
\mathcal{R}_4:\quad \q_4(z;\a_4,\b_4)
&= \frac{\left(\q_0'+\sqrt{-2\b_0}\right)^2+\left(4\a-4+2\sqrt{-2\b_0}\right)\q_0^2-\q_0^2(\q_0+2z)^2}{2\q_0\left(\q_0^2+2z\q_0+\q_0'+\sqrt{-2\b_0}\right)},
\end{align}\end{subequations} 
where $\q_0\equiv \q(z;\a_0,\b_0)$ and 
\begin{subequations}\begin{align}
(\a_1,\b_1)&=\left(\a_0+1,-\tfrac12 \big(2-\sqrt{-2\b_0}\big)^2\right),\\ 
(\a_2,\b_2)&=\left(\a_0-1,-\tfrac12 \big(2+\sqrt{-2\b_0}\big)^2\right),\\ 
(\a_3,\b_3)&=\left(\a_0+1,-\tfrac12 \big(2+\sqrt{-2\b_0}\big)^2\right),\\
(\a_4,\b_4)&=\left(\a_0-1,-\tfrac12 \big(2-\sqrt{-2\b_0}\big)^2\right).
\end{align}\end{subequations}
Fokas, Mugan and Ablowitz \cite{afumma88} also defined
the composite transformations $\mathcal{R}_5=\mathcal{R}_1\mathcal{R}_3$
and
$\mathcal{R}_7=\mathcal{R}_2\mathcal{R}_4$ given by
\begin{subequations}
\begin{align}\label{STR5} 
\mathcal{R}_5:\quad \q_5(z;\a_5,\b_5)
&=\frac{\left(\q_0'-\q_0^2-2z\q_0\right)^2+2\b_0}{2\q_0\left[ \q_0'-\q_0^2-2z\q_0+2\left(\a_0+1 \right)\right]},\\
\mathcal{R}_7:\quad \q_7(z;\a_7,\b_7)
&=-\,\frac{\left(\q_0'+\q_0^2+2z\q_0\right)^2+2\b_0}{2\q_0\left[\q_0'+\q_0^2+2z\q_0-2\left(\a_0-1 \right)\right]},\label{STR7} 
\end{align}
respectively, where 
\begin{equation}
(\a_5,\b_5)=(\a_0+2,\b_0),\qquad 
(\a_7,\b_7)=(\a_0-2,\b_0).
\end{equation}
\end{subequations}
We remark that $\mathcal{R}_5$ and $\mathcal{R}_7$ are the transformations
$\mathcal{T}_+$ and $\mathcal{T}_-$, respectively, given by Murata \cite{ym85}.

The application of $\mathcal{T}_1^+$ and $\mathcal{T}_2^+$ to the seed solution $\Sigma_0(z)$ given by (\ref{c1}) generates, in turn, new positive solutions
\begin{equation} \label{c5}
\begin{array}{l@{\qquad}l@{\qquad}l@{\qquad}l} \Sigma_{2}(z) = {1}/{z} , &\a_{1} = 2 ,& \b_{1} = -2 ,& z > 0, \\
\widetilde{\Sigma}_{2}(z) = -{1}/{z} ,& \a_{2} = -2 ,& \b_{2} = -2 ,& z < 0.\end{array} \end{equation}
Insertion of $\Sigma_{1}$ and $\Sigma_{2}$ in (\ref{b16}) where $1-\Lambda^2$ is determined by (\ref{b18}) with $z^*=\pm \tfrac12z^2$ delivers associated solutions of the \Ep\ IV system \eqref{b3} with $J\left({\Phi_2}/{\Phi_1}\right)$ given by (\ref{b19}).

The repeated action of the B\"{a}ck\-lund transformations $\left\{\mathcal{T}_j^\pm\right\}_{j=1}^4$
and other \bts\ such as the quartet $\left\{\mathcal{R}_j\right\}_{j=1}^4$
as constructed in \cite{afumma88} by the isomonodromy deformation method may be applied systematically to generate a multiplicity of exact solutions (\ref{b16}) of the \Ep\ IV systems \eqref{b3} which are here subject to the requirement $\Sigma>0$.
Thus, in particular, the repeated application of the B\"{a}ck\-lund transformation $\mathcal{R}_3$ \eqref{STR3} to $\Sigma_{2}$ as defined by (\ref{c5}) 
generates the following {additional} solutions of \PIV\ 
\begin{equation*}
\begin{array}{l@{\qquad}l@{\qquad}l}
\displaystyle\Sigma_3(z) = \frac{4z}{2z^2+1} ,& \a_3 = 3 ,& \b_3 = -8, \\[10pt]
\displaystyle\Sigma_4(z) = \frac{3(2z^2+1)}{z(2z^2+3)} ,& \a_4 = 4 ,& \b_4 = -18, \\[10pt]
\displaystyle\Sigma_5(z) = \frac{8z(2z^2+3)}{4z^4+12z^2+3} ,& \a_5 = 5 ,& \b_5 = -32,\\[10pt]
\displaystyle\Sigma_6(z) = \frac{5(4z^4+12z^2+3)}{z(4z^4+20z^2+15)} ,& \a_6 = 6 ,& \b_6 = -50,
\end{array}\end{equation*}
\hide{and application of the B\"{a}ck\-lund transformation $\mathcal{R}_2$ \eqref{STR2} to $\Sigma_{2}^-$ as defined by (\ref{c5}) 
generates a new solution of \PIV\ with
\begin{equation*}
\begin{array}{l@{\qquad}l@{\qquad}l}
\displaystyle\Sigma_3^-(z) = -\frac{4z}{2z^2-1} ,& \a_3 = -3 ,& \b_3 = -8, \\[10pt]
\displaystyle\Sigma_4^-(z) = -\frac{3(2z^2-1)}{z(2z^2-3)} ,& \a_4 = -4 ,& \b_4 = -18, \\[10pt]
\displaystyle\Sigma_5^-(z) = -\frac{8z(2z^2-3)}{4z^4-12z^2+3} ,& \a_5 = -5 ,& \b_5 = -32,\\[10pt]
\displaystyle\Sigma_6^-(z) = -\frac{5(4z^4-12z^2+3)}{z(4z^4-20z^2+15)} ,& \a_6 = -6 ,& \b_6 = -50,
\end{array}\end{equation*}}%
{These results constitute part of the catalogue of rational solutions of PIV set down in Table 3 of \cite{bch95} and some related forms have been noted in the context of supersymmetric quantum mechanics, see \cite{refMar09}.} Thus, the class of \Ep\ IV systems with $J\left(\Phi_2/\Phi_1\right)$ given by (\ref{b19}) admits exact solutions $\Phi_1$, $\Phi_2$ determined by the relations (\ref{b16}), 
where 
\begin{align*}
z_3^* &= \tfrac14\left\{{z^2} +  \ln z \right\},\\[2pt]
z_4^* &= \tfrac16\left\{{z^2} +  \ln (2z^2+1) \right\},\\[2pt]
z_5^* &= \tfrac18\left\{{z^2} +\ln [z(2z^2+3)]\right\},\\[2pt]
z_6^* &= \tfrac1{10}\left\{{z^2} +  \ln (4z^4+12z^2+3) \right\},
\end{align*}
for $z > 0$, so that from \eqref{b18}
\begin{equation*}
1 - \Lambda^2 = \Cases{\left[ \pm 2\sqrt{2}\, a (m - 1) \left( \tfrac12{z^2}+ \tfrac14 \ln z + \mathcal{C} \right) \right]^{2/(1-m)}, 
&\mbox{\rm if}\quad m \neq 1, \\
\exp \left\{ \pm 4\sqrt{2}\, a \left( \tfrac12{z^2} + \tfrac14 \ln z + \mathcal{C} \right) \right\}, &\mbox{\rm if}\quad m = 1.} 
\end{equation*}

\begin{lemma}For the rational solutions $\Sigma_n(z)$
\[ \int^z \frac{\d s}{\Sigma_{n}(s)}= \frac{1}{2(n-1)}\left\{z^2+\ln \deriv{h_{n-1}}{z}\right\},\]
where $h_{n-1}(z)=(-\i)^nH_{n-1}(-\i z)$, with $H_m(x)$ the Hermite polynomial.
\end{lemma}
\begin{proof}The rational solutions $\Sigma_n(z)$ have the form \cite{bch95}
\beq \Sigma_{n}(z)=\deriv{}{z}\ln h_{n-1}(z),\label{SigmaRat}\eeq
where $h_{n-1}(z)=(-\i)^nH_{n-1}(-\i z)$, with $H_m(x)$ the Hermite polynomial.
Since $h_{n-1}(z)$ satisfies the equation
\[\deriv[2]{h_{n-1}}{z}+2z\deriv{h_{n-1}}{z}-2(n-1) h_{n-1}=0,\]
which follows from the differential equation satisfied by Hermite polynomials (cf.~\cite[\S18.9]{refDLMF}),
then
\[ \frac{1}{\Sigma_{n}(z)}=\frac{h_{n-1}(z)}{h_{n-1}'(z)}=\frac{z}{n-1}+\frac{1}{2(n-1)}\frac{h_{n-1}''(z)}{h_{n-1}'(z)},\]
and so the result follows.\end{proof}

\begin{remark}{\rm
The rational solutions $\Sigma_{n}(z)$ given by \eqref{SigmaRat} satisfy the Riccati equation
\[\deriv{\Sigma_{n}}{z}+\Sigma_{n}^2+2z\Sigma_{n}-2(n-1)=0.\]
}\end{remark}

\section{Bound State Solutions}
The preceding describes how seed solutions $\Sigma_0(z)$ of \PIV\ \eqref{b8} can be used to generate associated classes of exact solutions of associated hybrid \Ep\ IV systems. The procedure is constrained, in general, by the requirement that attention be restricted to regions on which $\Sigma(z)$ is positive. {In this connection, it is remarked that the application of exact solutions of Painlev\'e equations to physical problems commonly requires consideration of restricted regions such as those associated with a positivity requirement. Thus, in \cite{lbjncrws10} in the application of a Painlev\'e II model to boundary value problems in two-ion electro-diffusion, the ion concentrations, which are necessarily positive, are governed by the classical Painlev\'e XXXIV equation which is directly related to the single component Ermakov-Painlev\'e II equation \cite{cr14}. In \cite{lbjncrws10}, regions were isolated wherein such positivity holds on half-space regions and boundary value problems were solved thereon. In related work similarity solutions of Painlev\'e II in terms of Yablonskii-Vorob'ev polynomials or classical Airy functions as generated by the iterated action of a B\"acklund transformation have been used recently to solve a range of moving problems for solitonic equations on regions $0<x<S(t)$ \cite{cr15,cr16,cr17}.} 

Bassom \etal\ \cite{bchm92} showed that if $\a$ is an odd integer and $\b=0$ then exact solutions of \PIV\ \eqref{b8} exist with the property that $\Sigma(z)\rightarrow0$ exponentially as $z\rightarrow\pm\infty$. In particular, it may be verified that
\begin{equation}
\label{bstate}
\Sigma_0(z)=\frac{\xi\exp(-z^2)}{\sqrt{\pi}\left[1-\tfrac12\xi{\erfc}(z)\right]},
\end{equation}
%
is a solution of \PIV\ \eqref{b8} corresponding to parameters $\a=1$ and $\b=0$. Here $\xi$ is a constant and $\erfc(z)$ denotes the usual complementary error function: it is evident that for $0<\xi<1$, the solution (\ref{bstate}) of \PIV\ \eqref{b8} is positive for all real $z$. It is now a routine matter to generate the associated class of exact solutions of the hybrid \Ep\ IV systems. More generally, if we denote the bound state solution of \PIV\ \eqref{b8} with $\a=2n+1$ and $\b=0$ as $\Sigma_n(z)$, then it can be proved that the \bt\
\begin{equation}\label{36}
\Sigma_{n+1}(z) = \frac{\big[\Sigma'_n(z) - \Sigma^2_n(z) - 2z \Sigma_n(z)\big]^2} {2\Sigma_n(z) \big[\Sigma'_n(z) - \Sigma^2_n(z) - 2z \Sigma_n(z) + 4(n + 1)\big]},\end{equation}
which is a special case of $\mathcal{R}_5$ \eqref{STR5},
can be applied iteratively to generate an infinite sequence of bound state solutions of \PIV\ \eqref{b8} commencing with the seed solution $\Sigma_0(z)$. If we denote $\psi(z)=1-\tfrac12\xi{\erfc}(z)$ so that  
\beq\label{37a}\Sigma_0(z)=\frac{\psi'(z)}{\psi(z)}:=\Psi(z),\eeq
then it may be shown that the next two of members of the bound state hierarchy of \PIV\ \eqref{b8} are given by
\begin{subequations}\label{37bc}\begin{align} \label{37b}
\Sigma_1(z) &= -\,\frac{\Psi(z) [\Psi(z) + 2z]^2} {\Psi^2(z) + 2z\Psi(z) - 2},\\
\label{37c}\Sigma_2(z) &= \frac{4\Psi (z)[\Psi^2(z) + 3z\Psi(z) + 2z^2 - 1]^2}{\big[\Psi^2(z) + 2z\Psi(z) - 2\big]\big[z\Psi^3(z) + (4z^2+3)\Psi^2(z) + 2z(2z^2+3)\Psi(z) - 4\big]},
\end{align}\end{subequations} with $\Psi(z)$ given by \eqref{37a}.
\hide{\begin{subequations}\label{37bc}\begin{align} \label{37b}
\Sigma_1(z) &= -\,\frac{\Psi(z) [\Psi(z) + 2z]^2} {\Psi^2(z) + 2z\Psi(z) - 2},\\
\label{37c}\Sigma_2(z) &= \frac{4\Psi(z) [\Psi^2(z) + 3z\Psi (z)+ 2z^2 - 1]^2}{\big[\Psi^2(z) + 2z\Psi(z) - 2\big]\big[z\Psi^3(z) + (4z^2+3)\Psi^2(z) + 2z(2z^2+3)\Psi(z) - 4\big]}.
\end{align}\end{subequations}}%
Each bound state solution $\Sigma_n(z)$ decays exponentially as $|z|\rightarrow\infty$ and has $n$ distinct zeros and satisfies
\beq\label{eq:wn}
\deriv[2]{\Sigma_n}{z} = \frac{1}{2\Sigma_n}\left(\deriv{\Sigma_n}{z}\right)^{2} + \frac{3}{2}\Sigma_n^3 + 4z \Sigma_n^2 + 2(z^2 - 2n-1)\Sigma_n, \eeq
which is \PIV\ \eqref{b8} with parameters $\a=2n+1$ and $\b=0$. 
These bound state solutions $\Sigma_n(z)$ have regions separated by the zeros on which they are positive. Thus, it is not only the fundamental bound state solution $\Sigma_0(z)$ of \PIV\ \eqref{b8} that is non-negative on the whole $z$-axis, but this property is shared by the higher members of the hierarchy. {These bound-state solutions are distinguished in that they decay exponentially and while many solutions of Painlev\'e IV are bounded (see e.g. \cite{refBerm12}), those such as the rational solutions do not decay as rapidly.}
Plots of $\Sigma_n(z;\xi)$, $n=1,2,\ldots,6$ for various $\xi$ are given in Figure \ref{fig:bstate1}.

\begin{lemma}\label{lemma1} For bound state solutions $\Sigma_{n}(z)$ of \PIV\ 
\begin{equation*}
\frac{1}{\Sigma_{n+1}(z)}=\frac{1}{\Sigma_{n}(z)}+\deriv{}{z}\left[ \frac{2}{\Sigma'_n(z) - \Sigma^2_n(z) - 2z \Sigma_n(z)}\right].
\end{equation*}
\end{lemma}
\begin{proof}By definition
\begin{align*}
\deriv{}{z}\left[ \frac{2}{\Sigma'_n(z) - \Sigma^2_n(z) - 2z \Sigma_n(z)}\right]
&=-\frac{2\left\{\Sigma''_n(z) - 2[\Sigma_n(z)+z]\Sigma_n'(z) -2\Sigma_n\right\}}{\left[\Sigma'_n(z) - \Sigma^2_n(z) - 2z \Sigma_n(z)\right]^2}\\
&= \frac{8(n+1)\Sigma_n}{\left[\Sigma'_n(z) - \Sigma^2_n(z) - 2z \Sigma_n(z)\right]^2}\\ &\qquad
+\frac{2\Sigma_n}{\Sigma'_n(z) - \Sigma^2_n(z) - 2z \Sigma_n(z)}-\frac{1}{\Sigma_n(z)},
\end{align*}
since $\Sigma_n(z)$ satisfies \eqref{eq:wn}.
Further from the \bt\ \eqref{36} we have
\begin{align*}
\frac{1}{\Sigma_{n+1}(z)} &= \frac{2\Sigma_n(z) \big[\Sigma'_n(z) - \Sigma^2_n(z) - 2z \Sigma_n(z) + 4(n + 1)\big]}{\big[\Sigma'_n(z) - \Sigma^2_n(z) - 2z \Sigma_n(z)\big]^2}\\
&= \frac{8(n+1)\Sigma_n(z)}{\left[\Sigma'_n(z) - \Sigma^2_n(z) - 2z \Sigma_n(z)\right]^2}
+\frac{2\Sigma_n(z)}{\Sigma'_n(z) - \Sigma^2_n(z) - 2z \Sigma_n(z)},
\end{align*}
and so the result follows.
\end{proof}

\begin{corollary} For bound state solutions $\Sigma_{n}(z)$ of \PIV\ 
\begin{equation}\label{intSn}
\int^z\frac{\d s}{\Sigma_{n+1}(s)}=\int^z\frac{\d s}{\Sigma_{n}(s)} +\frac{2}{\Sigma'_n(z) - \Sigma^2_n(z) - 2z \Sigma_n(z)}.
\end{equation}
\end{corollary}
\begin{proof}The result immediately follows from Lemma \ref{lemma1}. 
\end{proof}

\begin{lemma}For the bound state solution $\Sigma_{0}(z)$
 \[ \int^z\frac{\d s}{\Sigma_{0}(s)} = \frac{2-\xi}{2\xi}\sqrt{\pi} \int^z\exp(s^2)\,\d s
 +\sum_{n=0}^{\infty}\frac{n!\,2^{2n}z^{2n+2}}{(2n+2)!}.\]
\end{lemma}
\begin{proof}From \eqref{bstate} we have
\begin{align*}
\frac{1}{\Sigma_0(z)}&=\frac{\sqrt{\pi}\exp(z^2)}{\xi}-\tfrac12\sqrt{\pi} \exp(z^2)\, {\erfc}(z)\\
&=\frac{2-\xi}{2\xi}\sqrt{\pi} \exp(z^2)+\sum_{n=0}^{\infty}\frac{n!\,2^{2n}z^{2n+1}}{(2n+1)!},
\end{align*}
since
 \[\erfc(z)=1-\frac{2}{\sqrt{\pi}}\exp(-z^2)\sum_{n=0}^{\infty}\frac{n!\,2^{2n}z^{2n+1}}{(2n+1)!},\]
 see \cite[equation (7.6.2)]{refDLMF}. Hence the result follows. \end{proof}
  
 \begin{remark}{\rm We remark that
\[\erfc(z)=\frac{2}{\pi} \exp(-z^2)\int_{0}^{\infty}\frac{\exp\left(-z^2t^2\right)}{t^{2}+1}\,\d t,\]
see \cite[equation (7.7.1)]{refDLMF}, and
\[\erfc(z)=\pi^{-1/2} \exp(-z^2) U(\tfrac12,\tfrac12,z^2),\]
with $U(a,b,z)$ the Kummer function, see \cite[equation (7.11.5)]{refDLMF}. Also $S_0(z)=1/\Sigma_0(z)$ satisfies the equation
\[ \deriv[2]{S_0}{z}-2z\deriv{S_0}{z}-2S_0=0.\]}\end{remark}

From \eqref{37bc} and \eqref{intSn} we obtain
\begin{align*}
\int^z\frac{\d s}{\Sigma_{1}(s)}&=\int^z\frac{\d s}{\Sigma_{0}(s)} -\frac{1}{\Psi(z)[\Psi(z)+2z]},\\
\int^z\frac{\d s}{\Sigma_{2}(s)}&=\int^z\frac{\d s}{\Sigma_{0}(s)} -\frac {\Psi^{3}(z)+2z\Psi^{2}(z)+4z}{4\Psi(z)[\Psi ^{2}(z)+3z\Psi (z)+2z^{2}-1]},
\end{align*}
where $\Psi(z)$ is defined by \eqref{37a}.

{\def\plot#1{\includegraphics[width=2.25in,height=2.25in]{#1}}
\begin{figure}[ht]
\[ \begin{array}{@{\qquad}c@{\qquad}c}
\plot{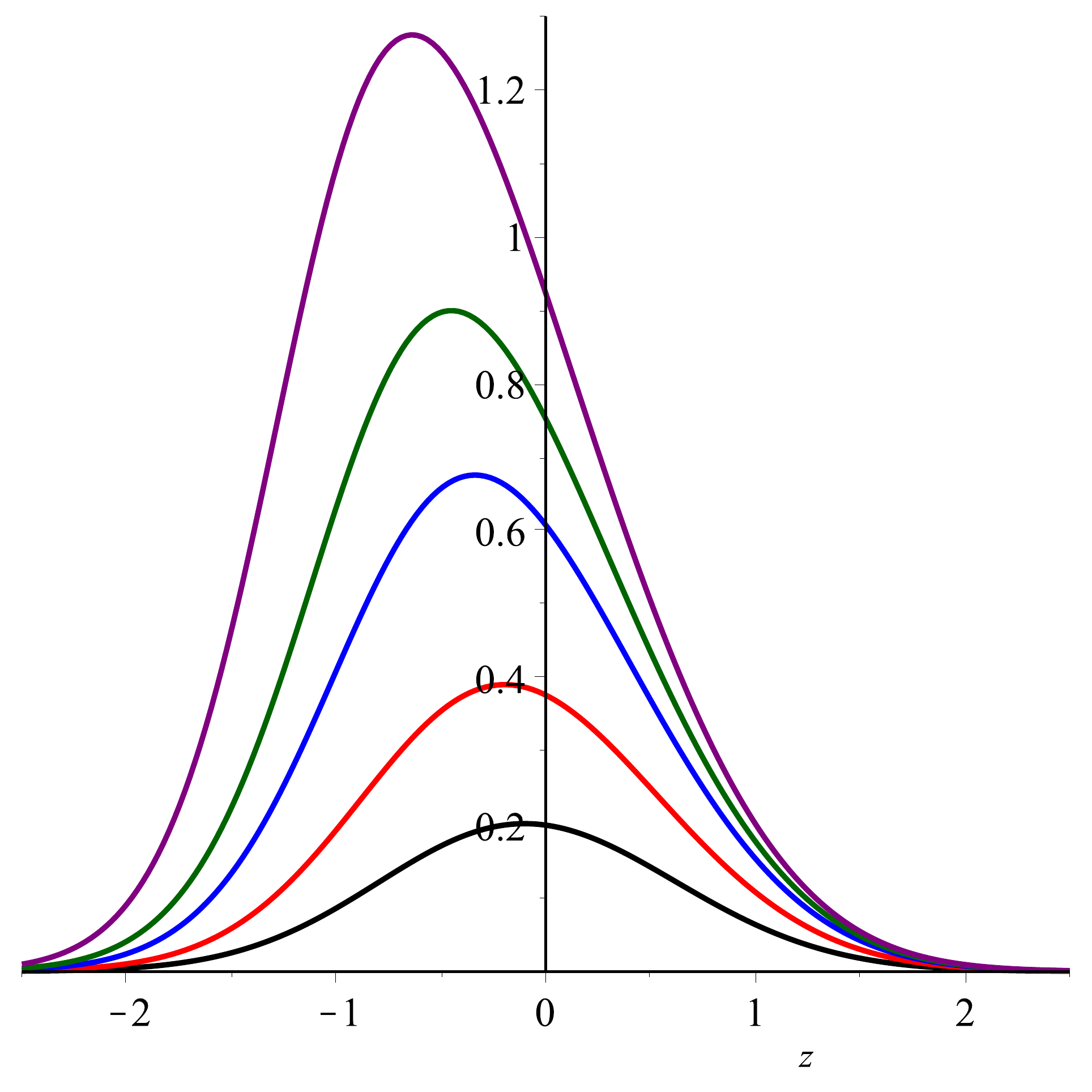} & \plot{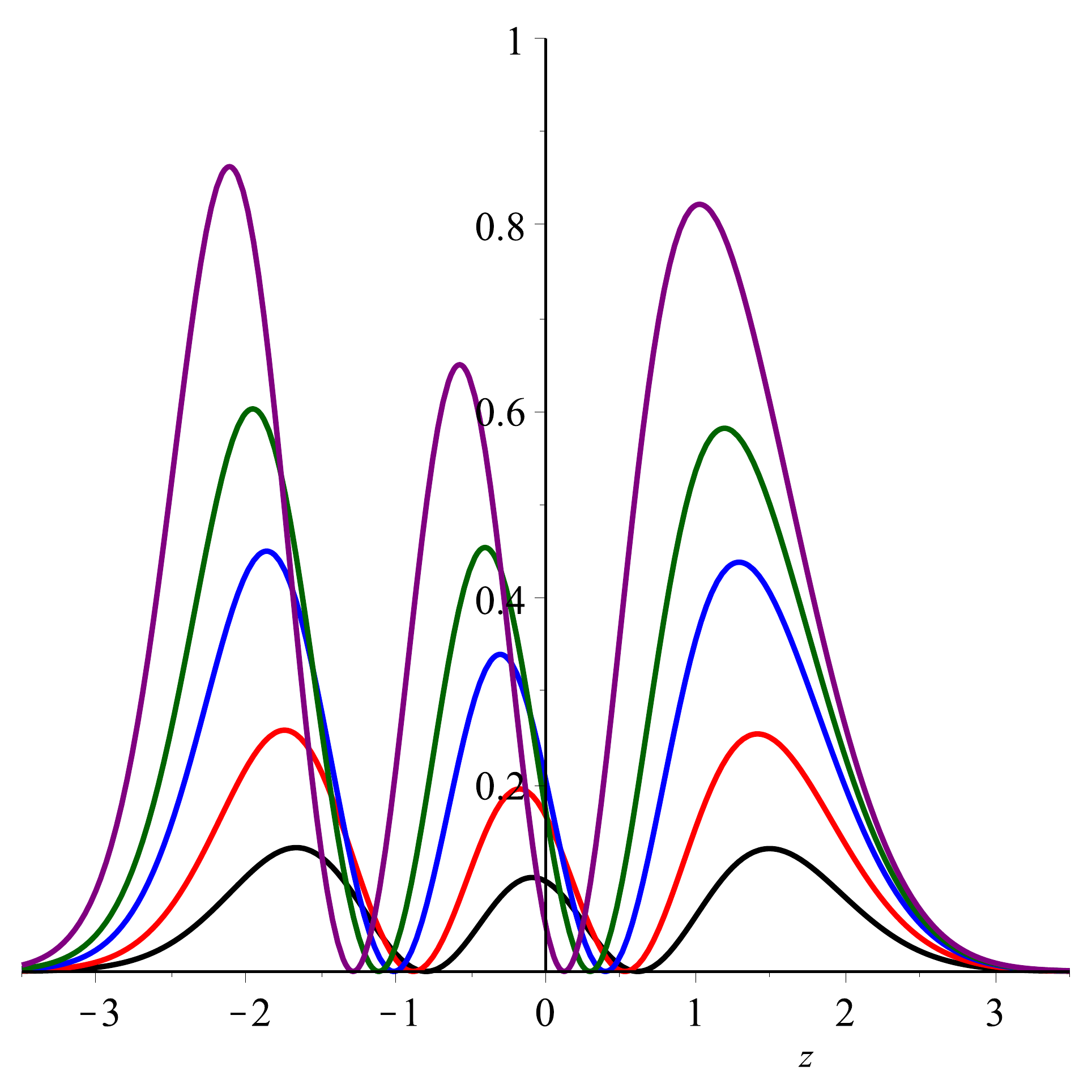}\\
\Sigma_1(z;\xi) & \Sigma_2(z;\xi)\\[7.5pt]
\plot{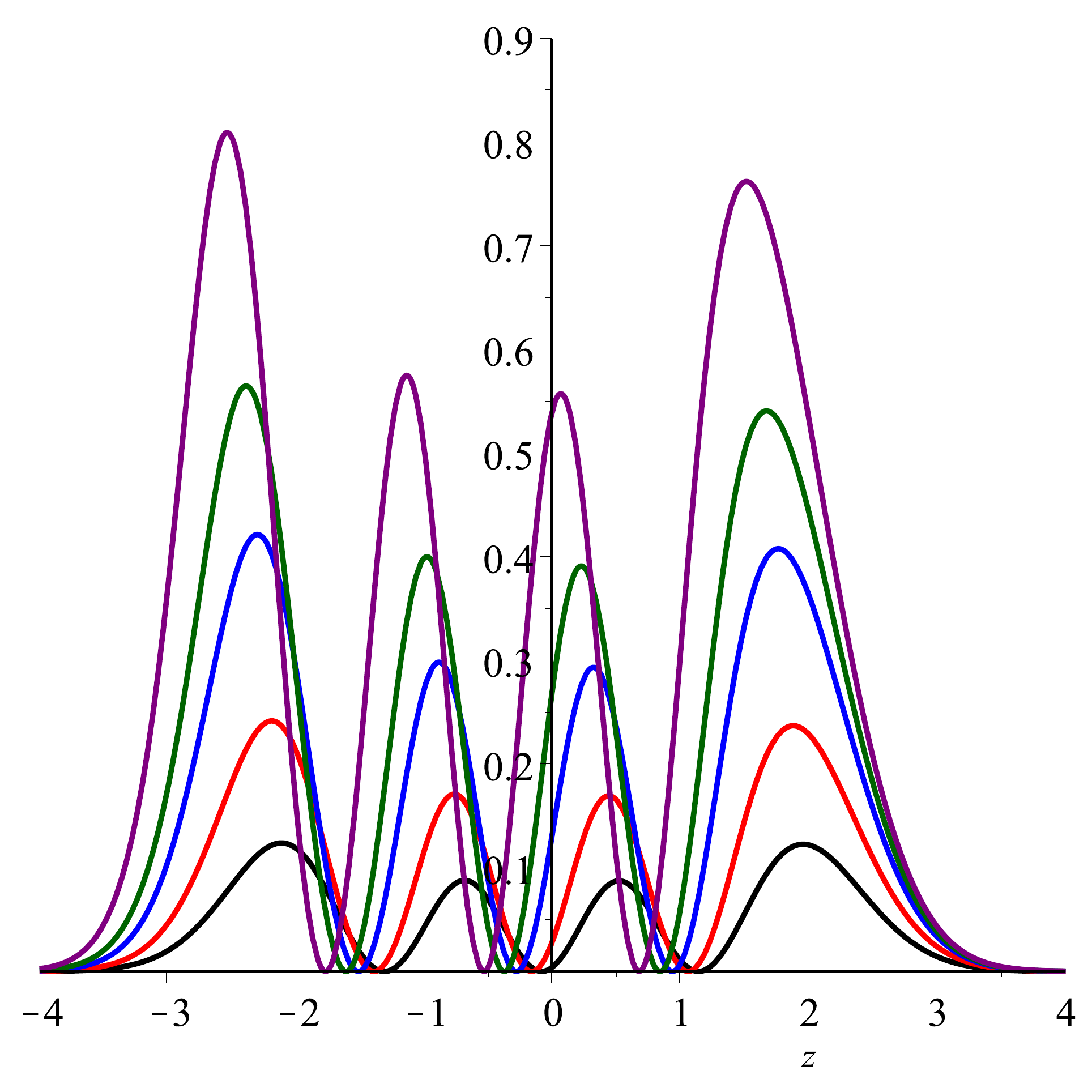} & \plot{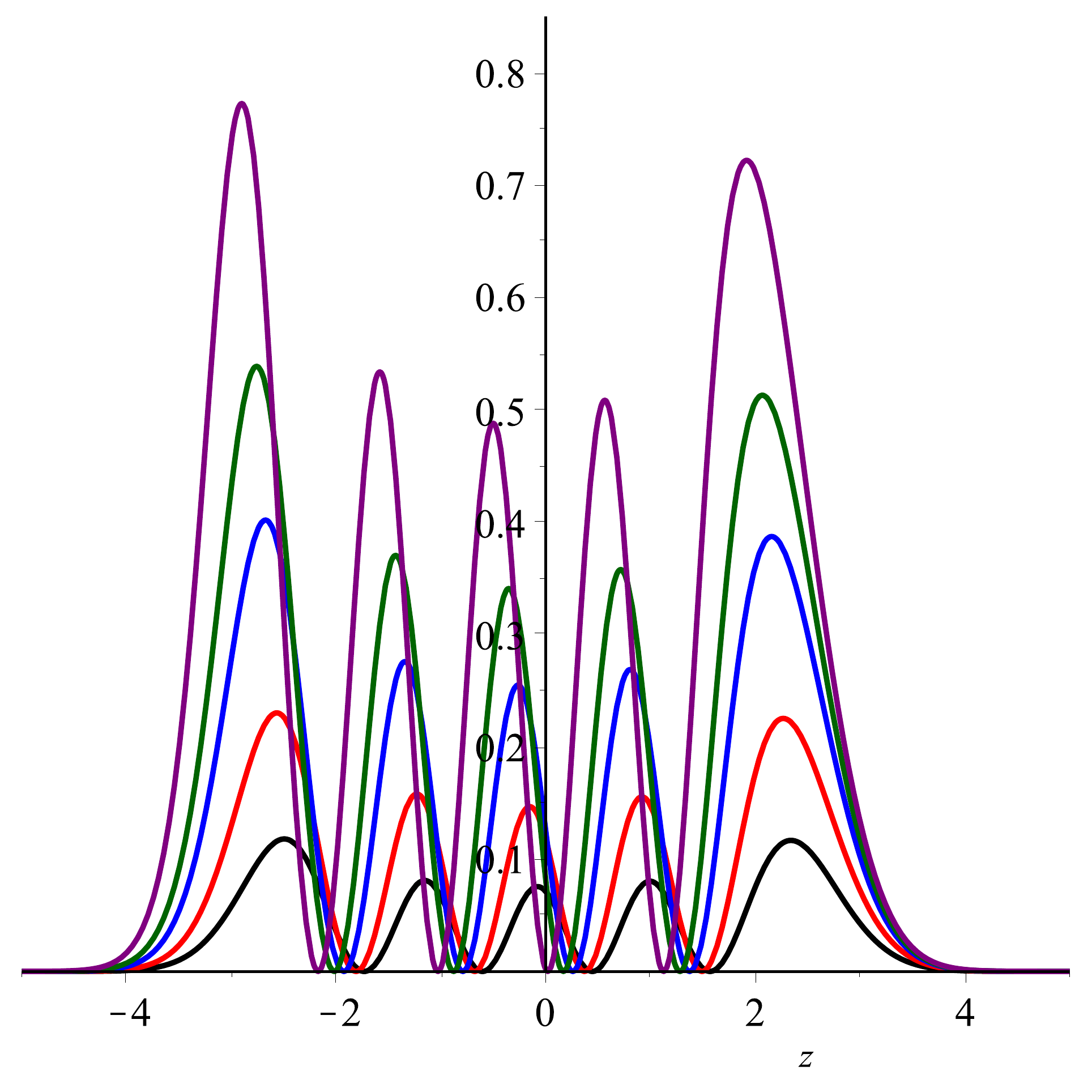}\\
\Sigma_3(z;\xi) & \Sigma_4(z;\xi)\\[7.5pt]
\plot{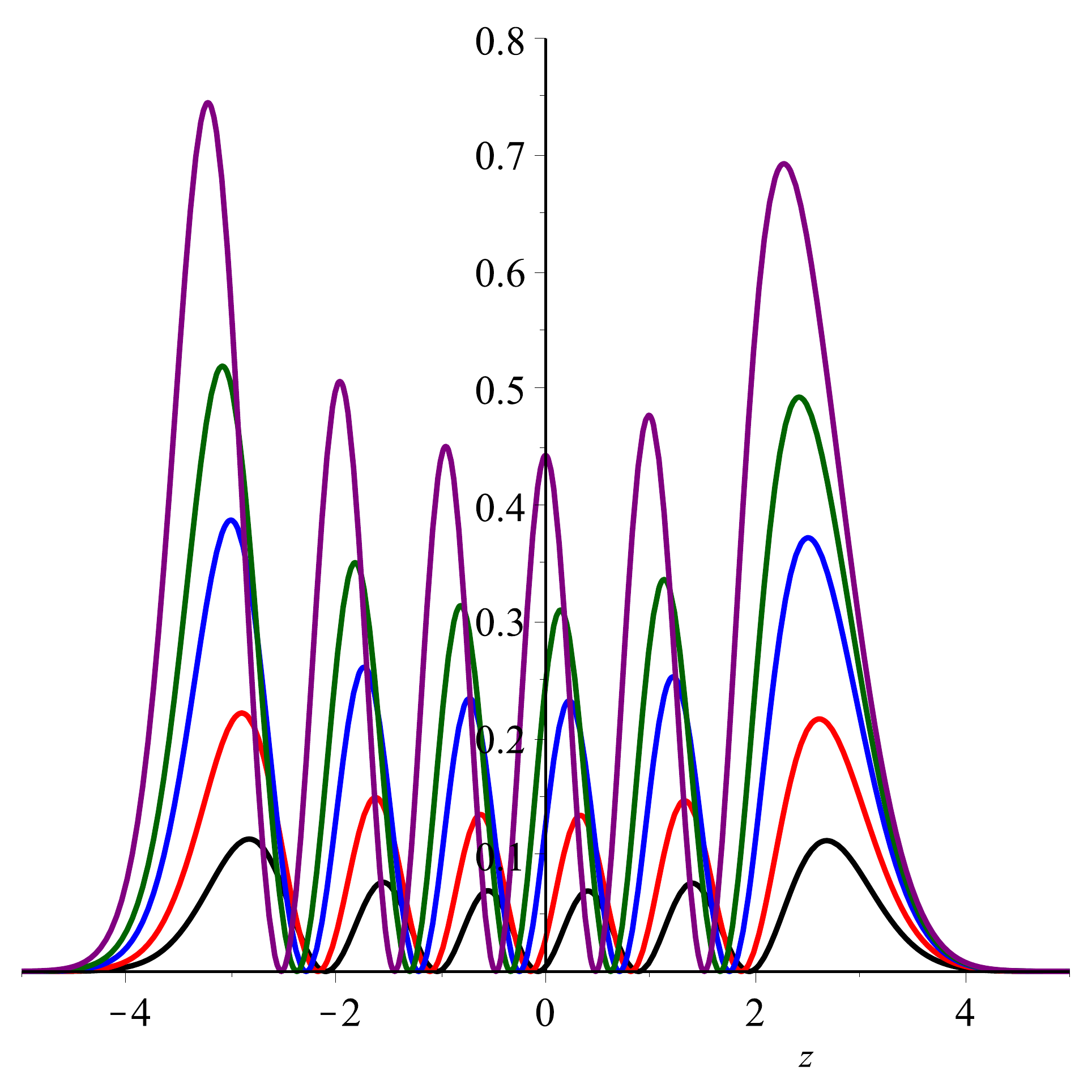} & \plot{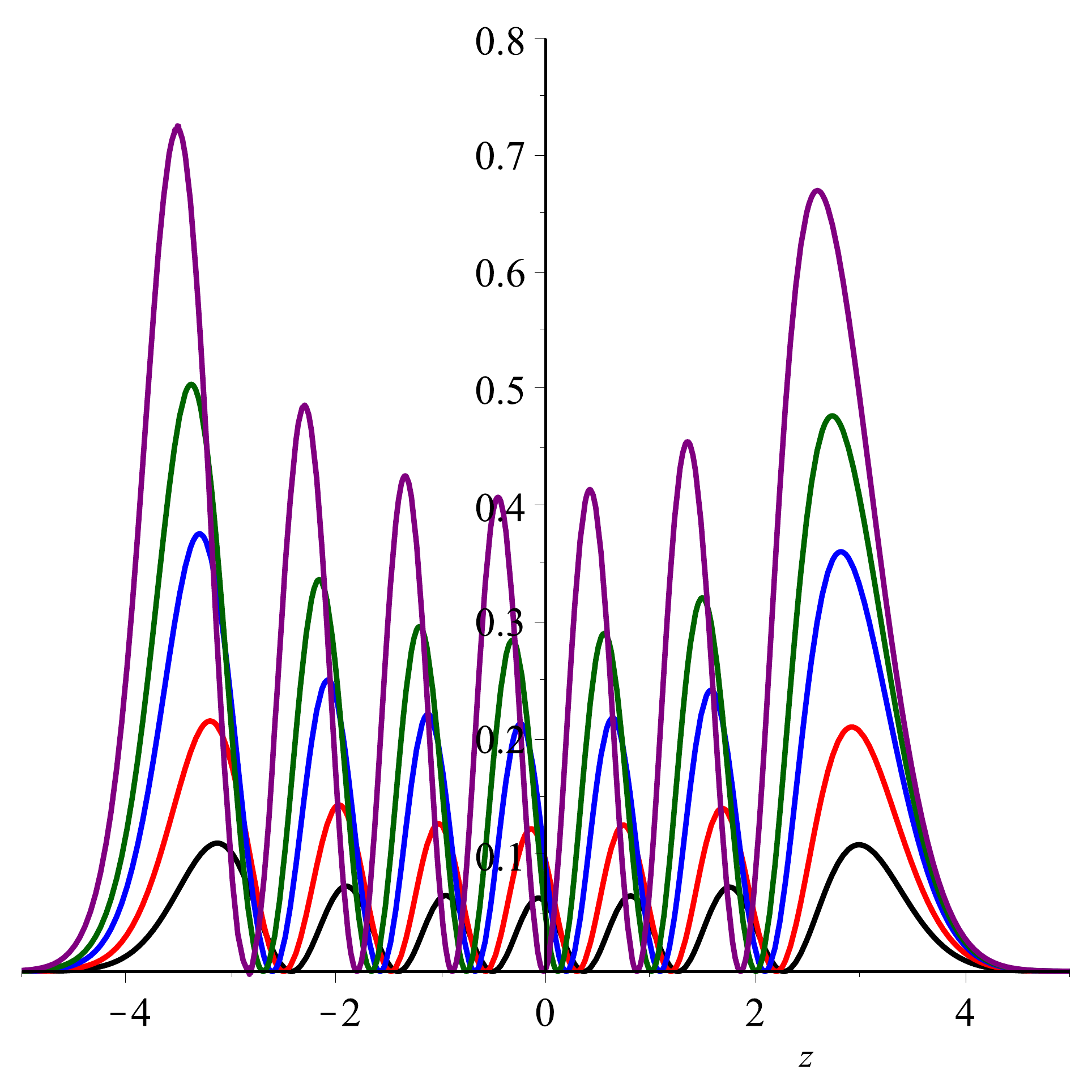}\\
\Sigma_5(z;\xi) & \Sigma_6(z;\xi)
\end{array}\]
\caption{\label{fig:bstate1}Plots of $\Sigma_n(z;\xi)$, $n=1,2,\ldots,6$ for $\xi=0.3$ (black), $0.5$ (red), $0.7$ (blue), $0.8$ (green), $0.9$ (purple).}
\end{figure}
\begin{figure}[ht]
\[ \begin{array}{c@{\qquad}c}
\plot{PIVbssol1} & \plot{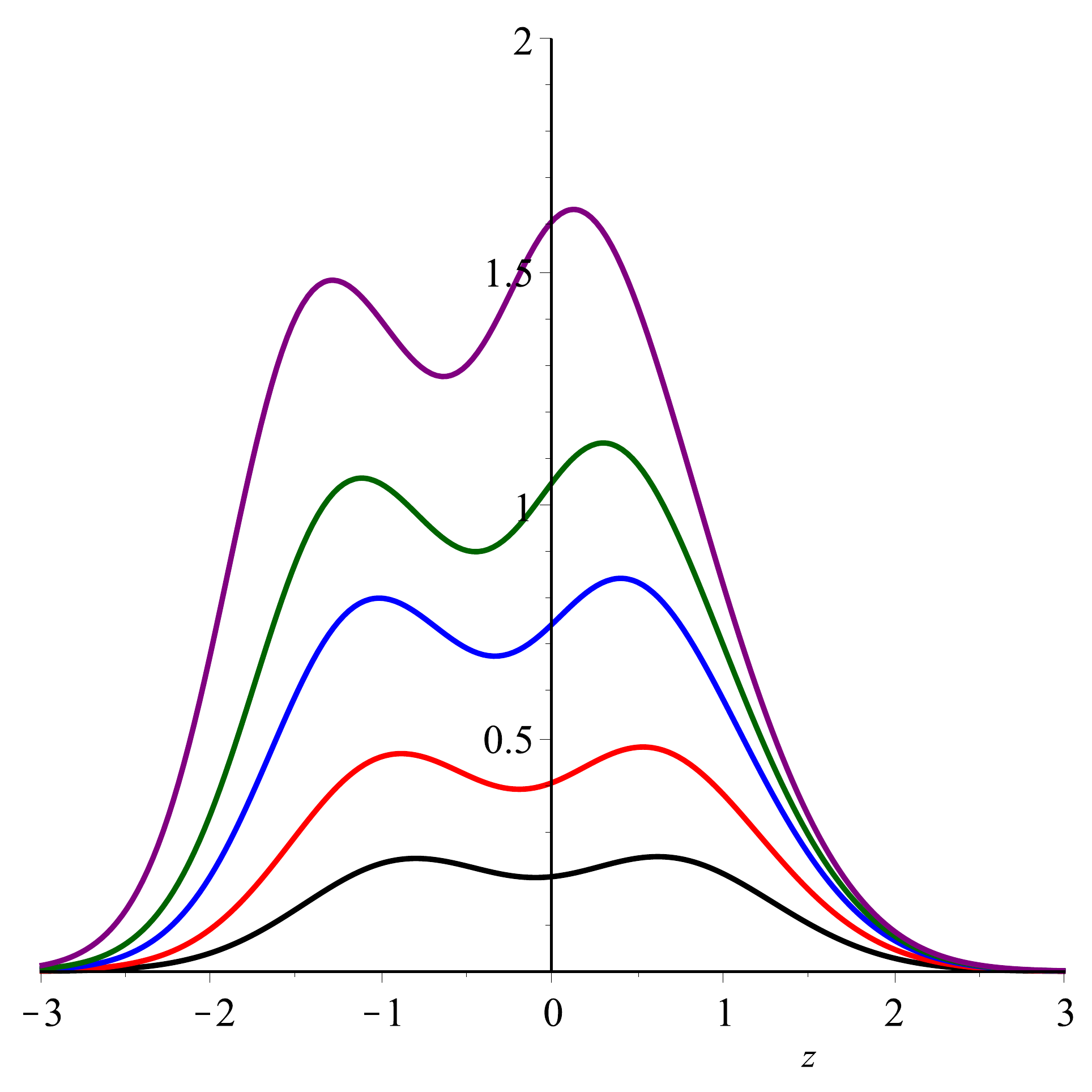}\\
\sigma_1(z;\xi) & \sigma_2(z;\xi)\\[7.5pt]
\plot{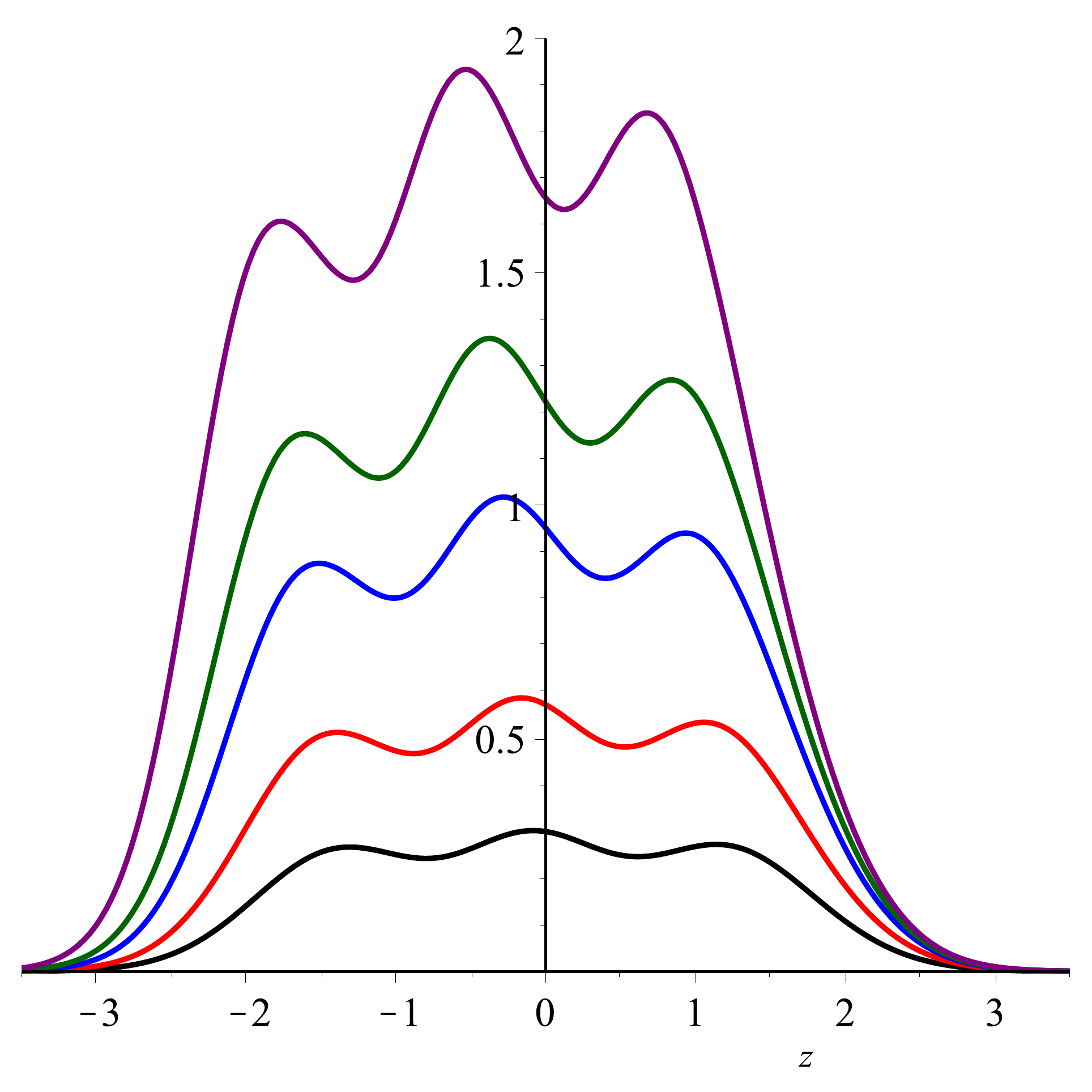} & \plot{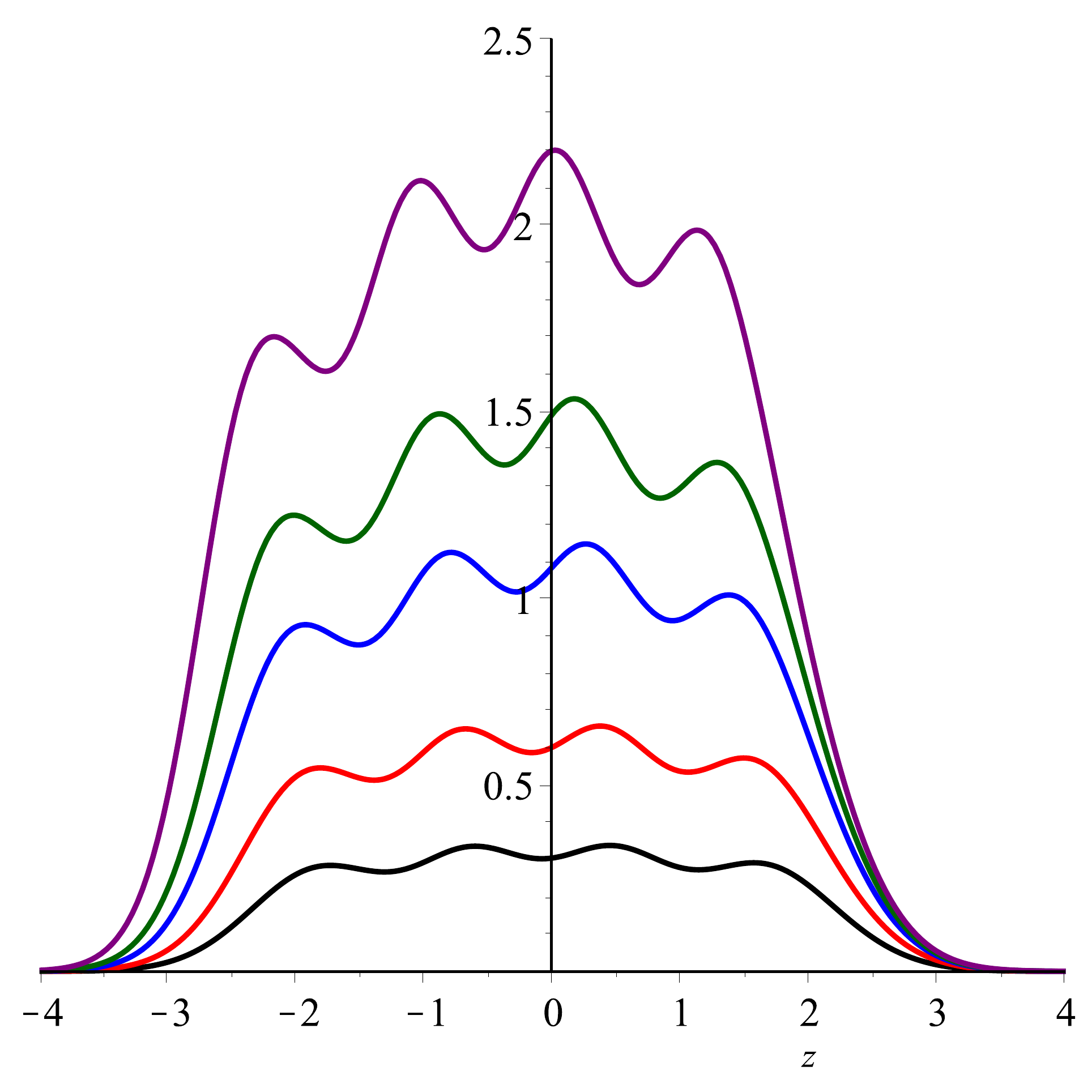}\\
\sigma_3(z;\xi) & \sigma_4(z;\xi)\\[7.5pt]
\plot{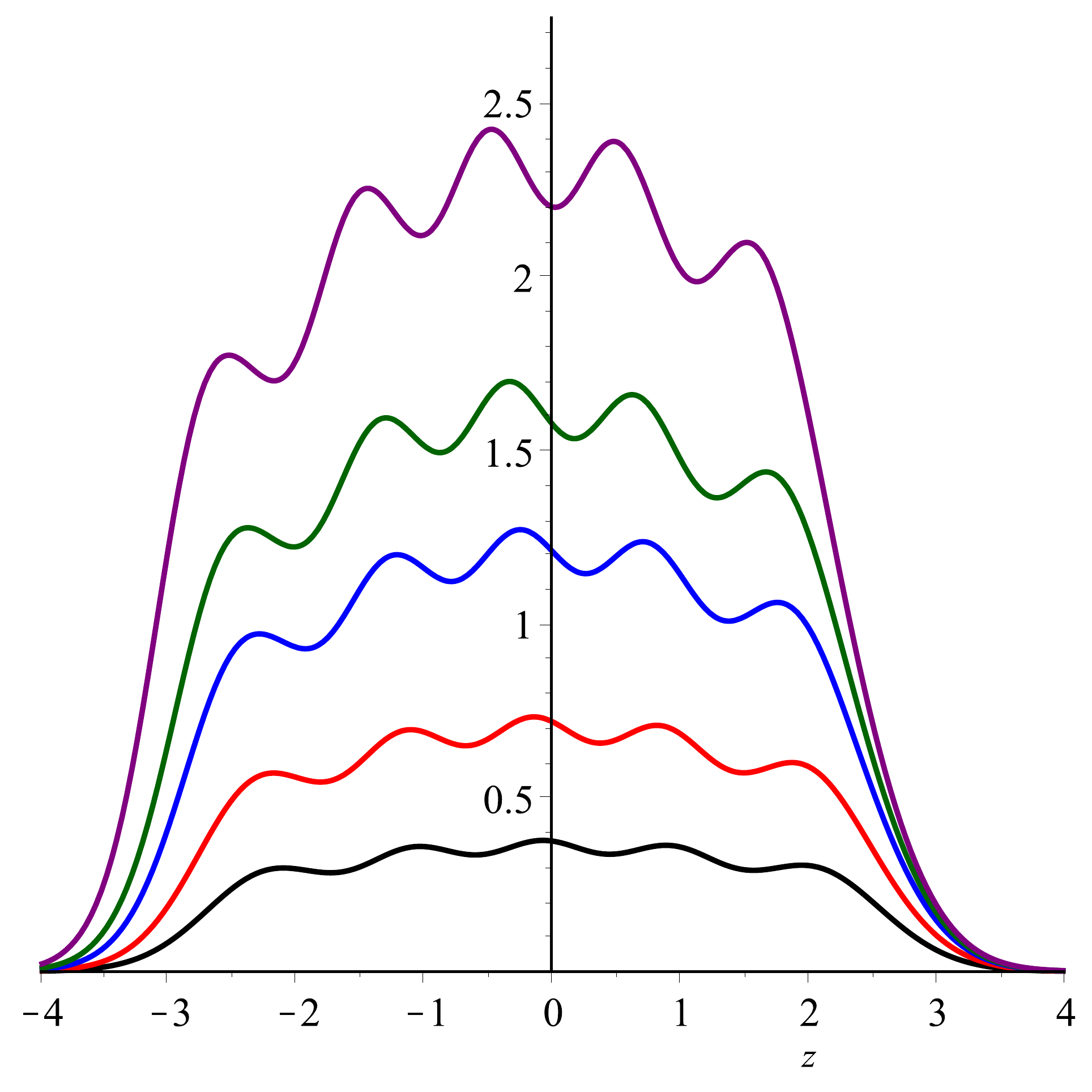} & \plot{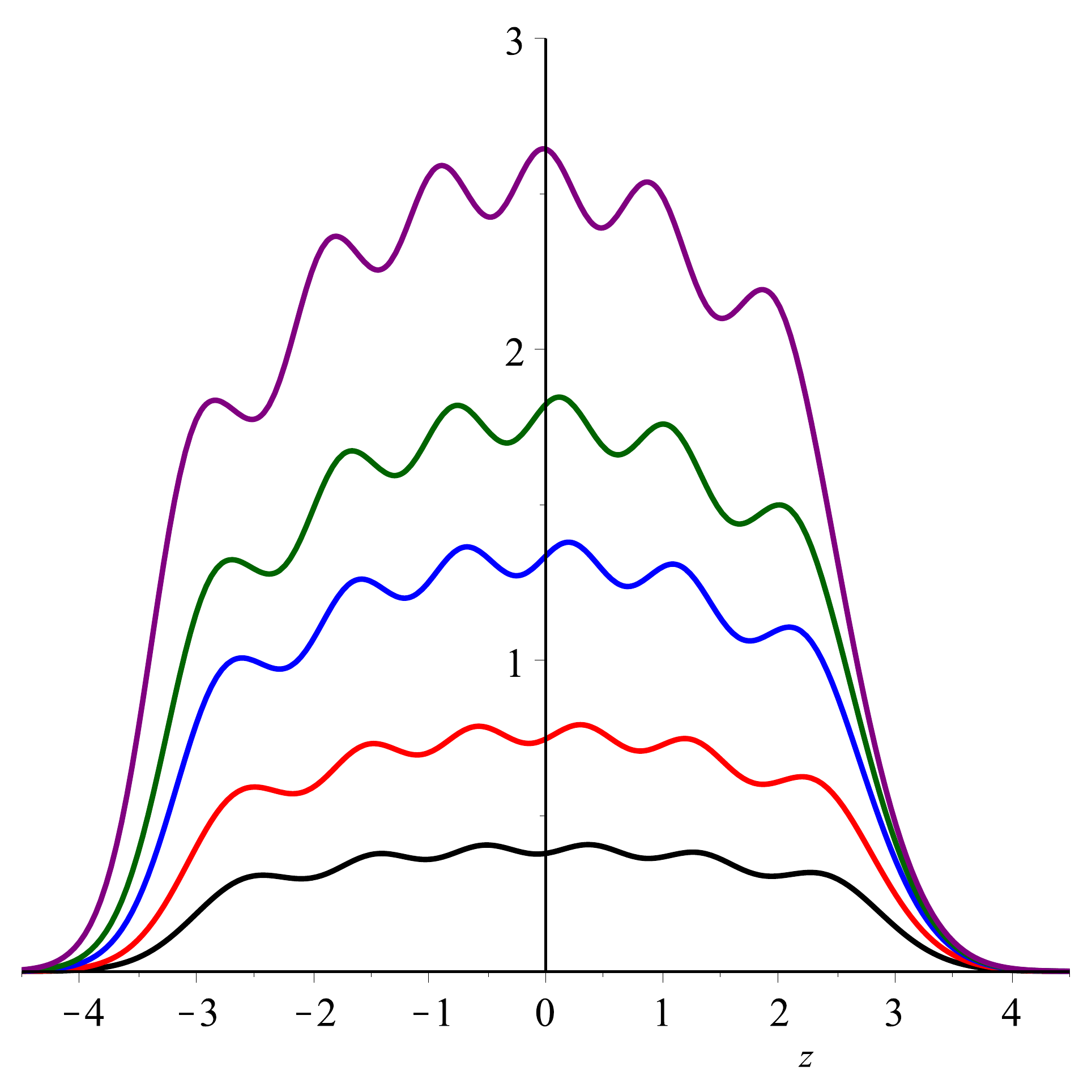}\\
\sigma_5(z;\xi) & \sigma_6(z;\xi)
\end{array}\]
\caption{\label{fig:bstate2}Plots of $\sigma_n(z;\xi)$, $n=1,2,\ldots,6$ for $\xi=0.3$ (black), $0.5$ (red), $0.7$ (blue), $0.8$ (green), $0.9$ (purple).}
\end{figure}
}

The bound state solutions $\Sigma_{n}(z)$ of \PIV\ can be given in terms of  Hankel determinants as described in the following theorem.
\begin{theorem}
Define $\tau_n(z;\xi)$ to be the Hankel determinant
\beq\label{def:taun}
\tau_n(z;\xi)=\det\Big[ \ph_{j+k}(z;\xi)\Big]_{j,k=0}^{n-1},\eeq
with $\xi$ a parameter and 
where $\ph_m(z;\xi)$ is given by
\beq\label{def:phin}
\ph_m(z;\xi) = (-\tfrac12\i)^{m}H_m(-\i z)+(-\tfrac12)^{m+1}\xi\exp(-z^2)\deriv[m]{}{z}\left\{\erfc(z)\exp(z^2)\right\},\eeq
with $H_m(z)$ the \textit{Hermite polynomial} and $\erfc(z)$ the \textit{complementary error function}. 
Then
\beq \label{def:wn}
\Sigma_n(z;\xi)= \deriv{}{z} \ln\frac{\tau_{n+1}(z;\xi)}{\tau_n(z;\xi)},\eeq
satisfies \eqref{eq:wn}.
\end{theorem}
\begin{proof}See Forrester and Witte \cite{refFW03}. \end{proof}

\begin{remark}{\rm Let $\tau_n(z;\xi)$ be the Hankel determinant given by (\ref{def:taun}), then
\beq\label{def:sigman}\sigma_n(z;\xi)=\deriv{}{z} \ln\tau_n(z;\xi),\eeq
satisfies
\[
\left(\deriv[2]{\sigma_n}z\right)^{2} - 4\left(z\deriv{\sigma_n}z-\sigma_n\right)^{2} +4\left(\deriv{\sigma_n}z\right)^2\left(\deriv{\sigma_n}z+2n\right)=0,\]
which is \SIV\ \eqref{eq:Siv} with parameters $\th_0=n$ and $\th_\infty=0$.
The solutions $\Sigma_n(z;\xi)$, given by (\ref{def:wn}), and $\sigma_n(z;\xi)$, given by (\ref{def:sigman}), are related as follows
\[ \Sigma_n(z;\xi) = \sigma_{n+1}(z;\xi)-\sigma_n(z;\xi).\]
Plots of $\sigma_n(z;\xi)$, $n=1,2,\ldots,6$  for various $\xi$ are given in Figure \ref{fig:bstate2}.
Further $\tau_n(z;\xi)$ given by (\ref{def:taun}) 
satisfies 
\[
\frac{\tau_{n+1}(z;\xi)\,\tau_{n-1}(z;\xi)}{\tau_n^2(z;\xi)}=\deriv[2]{}{z} \ln\tau_n(z;\xi) + 2n,\]
which is equivalent to the Toda equation.
}\end{remark}
\begin{lemma}The solutions $\Sigma_n(z;\xi)$, given by (\ref{def:wn}), and $\sigma_n(z;\xi)$, given by (\ref{def:sigman}), are positive and bounded provided that $0<\xi<1$.
\end{lemma}

\begin{remark}{\rm Since the {Hermite polynomial} $H_m(z)$ has the integral representation 
\[H_m(z)=\frac{2^m}{\sqrt{\pi}}\int_{-\infty}^{\infty}(z+\i t)^m\exp(-t^2)\,\d t,\]
and the {complementary error function} $\erfc (z)$ has the integral representation
\[ 
\erfc (z)=\frac{2}{\sqrt{\pi}} \int_z^\infty \exp (-t^2) \,\d t,\]
then it can easily be shown that $\ph_m(z;\xi)$ given by (\ref{def:phin}) has the integral representation
\beq\label{def:phin2}
\ph_m(z;\xi) =\frac{1}{\sqrt{\pi}}\int_{-\infty }^{\infty }\left[ 1-\xi \mathcal{H}(t-z) \right](t-z)^{m} \exp(-t^2)\,\d t,
\eeq
with $\mathcal{H}(x)$ the \textit{Heaviside function}.
}\end{remark}

These bound state solutions arise in the theory of:
\begin{itemize} 
\item[(i)] orthogonal polynomials with the discontinuous Hermite weight 
\begin{equation*}
w(x;z,\mu)=\exp(-x^2)\left\{1-\mu +2 \mu\mathcal{H}(x-z)\right\},\end{equation*}
with $\mathcal{H}(x)$ the Heaviside function and $\mu$ a parameter, see 
\cite{refCP05}; and 
\item[(ii)] GUE random matrices which are expressed as Hankel determinants of the function $\ph_m(z;\xi)$ given by (\ref{def:phin2}), see
\cite{refFW03}.
\end{itemize}

{\begin{remark}{\rm The bound state solutions discussed here 
are members of a general family of solutions of \PIV\ \eqref{eq:PIV}  with $\b=0$ and the boundary condition
\[ \omega(z)\to 0,\qquad\text{as}\quad z\to\infty,\]
which have been studied by various authors
\cite{bch93,bchm92,refIK98,refWZ09}
}\end{remark}}
\section{Conclusion}

Here, integrable systems of a novel hybrid \Ep\ IV kind have been introduced which both admit a Ermakov invariant and are connected to \PIV\ \eqref{b8}. In terms of applications, it is noted that particular such nonlinear coupled systems arise out of symmetry reduction of coupled derivative nonlinear \sch\ systems \cite{cr2014}. Hybrid integrable \Ep\ II systems with genesis in the context of three-ion electrodiffusion and the Nernst-Planck system have been recently developed in \cite{crws16,crws2016}. This suggests that a comprehensive investigation be undertaken into integrable \Ep\ systems which both possess characteristic Ermakov invariants and allow the construction of classes of exact solutions via \bts\ admitted by the canonical \PII-\PVI\ equations. {In most recent work \cite{cr2017}, an integrable Ermakov-Painlev\'e III system has been constructed in a manner directly analogous to that used to obtain integrable hybrid Ermakov-Painlev\'e II and Ermakov-Painlev\'e IV systems. Importantly, these Ermakov-Painlev\'e systems may be now embedded in multi-component Ermakov-Painlev\'e and Ermakov-Toda lattice schemes as introduced in \cite{crws17}. This method of construction of the hybrid Ermakov-Painlev\'e systems has potential specific application yet to be undertaken for residual Ermakov-Painlev\'e V and Ermakov-Painlev\'e VI systems.}

{\section*{Acknowlegement}We thank the reviewers for their helpful comments.}

 \def\bibitm{\vspace{-0.2cm}\bibitem}
\def\refpp#1#2#3#4#5{\bibitm{#1} \textrm{\frenchspacing#2}\ #5\ \textrm{#3} #4}

\def\refjl#1#2#3#4#5#6#7{\bibitm{#1} \textrm{\frenchspacing#2}, \textrm{#3},
{\frenchspacing\it#4}, \textbf{#5}\ (#7) #6.}

\def\refbk#1#2#3#4#5{\bibitm{#1} \textrm{\frenchspacing#2}, \textit{#3}, #4, #5.}

\def\refcf#1#2#3#4#5#6#7{\bibitm{#1} \textrm{\frenchspacing#2}, {#3},
 \textit{#4}, {\frenchspacing#5}\ #7, #6}

\def\JPA{J. Phys. A}
\def\CUP{Cambridge University Press}

\end{document}